\documentclass[a4paper,12pt]{article}
\usepackage[numbers]{natbib}
\usepackage{amsthm}
\usepackage{amsmath}
\usepackage{amsthm}
\usepackage{amssymb}
\usepackage{hyperref}
\usepackage{ifthen}
\usepackage{natbib}
\usepackage{color}
\usepackage{graphicx}
\usepackage{refcount}
\usepackage{graphicx,comment}%
\usepackage{tabu,algorithm,algpseudocode}%
\usepackage{setspace}
\usepackage{pdfpages}

\newtheorem{theorem}{Theorem}
\newtheorem*{theorem*}{Theorem}
\newtheorem{lemma}[theorem]{Lemma}

\begin{document}

\interfootnotelinepenalty=10000

\begin{titlepage}
\begin{center}


\textsc{BAR-ILAN UNIVERSITY}\\[2.5cm]


{ \huge \bfseries Two-Body Assignment Problem in the Context of the Israeli Medical Internship Match \\[0.5cm] }

M.Sc. Thesis\\[2.5cm]

\begin{minipage}[t]{0.4\textwidth}
\begin{flushleft} \large
SLAVA BRONFMAN
\end{flushleft}
\end{minipage}%
\begin{minipage}[t]{0.4\textwidth}
\end{minipage}\\[3.5cm]

Submitted as partial fulfillment of the requirements for the Masterӳ Degree in the Department of Computer Science, Bar-Ilan University.
\\[3.5cm]
\end{center}

\noindent Ramat Gan, Israel\hfill 2015 

\end{titlepage}

\thispagestyle{empty} 
This work was carried out under the supervision of Dr.Avinatan Hassidim from Department of Computer Science, Bar-Ilan University.
\newpage
\thispagestyle{empty} 

\section*{Acknowledgment}
This research project would not have been possible without the support of many people who have supported, motivated and encouraged me throughout my graduate studies.
\\\\
First and foremost, I would like to thank my advisor Dr. Avinatan Hassidim for giving me the opportunity to carry out this research under his professional and inspiring supervision. His academic guidance, constructive criticism, insightful advices, understanding, patience and warm personal character had a big influence on the outcome presented here and also on my own personal enlightening experience as a researcher.
\\\\
Deepest gratitudes are also due to Assaf Romm for all the time and efforts he spent assisting in this thesis, Assaf was abundantly helpful and offered invaluable assistance, guidance and support with the professional material.
\\\\
A lot of thanks to Prof.Noga Alon for the time and efforts he spent assisting in this thesis.
\\\\
I owe special thanks to my family and friends who provided me with encouragement all this time.


\newpage

\thispagestyle{empty} 
\renewcommand*\contentsname{Table of Contents}
\tableofcontents
\thispagestyle{empty} 
\newpage

\thispagestyle{empty} 
\listoffigures
\newpage

\thispagestyle{empty} 
\listoftables
\newpage

\pagenumbering{Roman} 

\begin{abstract}
The final step in getting an Israeli M.D. is performing a year-long internship in one of the hospitals in Israel. Internships are decided upon by a lottery, which is known as ``The Internship Lottery''. In 2014 we redesigned the lottery, replacing it with a more efficient one. The new method is based on calculating a tentative lottery, in which each student has some probability of getting to each hospital. Then a computer program ``trades'' between the students, where trade is performed only if it is beneficial to both sides. This trade creates surplus, which translates to more students getting one of their top choices. The average student improved his place by $0.91$ seats. The new method can improve the welfare of medical graduates, by giving them more probability to get to one of their top choices. It can be applied in internship markets in other countries as well.

This thesis presents the market, the redesign process and the new mechanism which is now in use. There are two main lessons that we have learned from this market. The first is the ``Do No Harm'' principle, which states that (almost) all participants should prefer the new mechanism to the old one. The second is that new approaches need to be used when dealing with two-body problems in object assignment. We focus on the second lesson, and study two-body problems in the context of the assignment problem. We show that decomposing stochastic assignment matrices to deterministic allocations is NP-hard in the presence of couples, and present a polynomial time algorithm with the optimal worst case guarantee. We also study the performance of our algorithm on real-world and on simulated data.
\end{abstract}
\newpage

\pagenumbering{arabic} 

\section{Introduction} \label{sec:introduction}
Prior to receiving their medical degrees Israeli medical graduates (and foreign trained doctors) must participate in a year-long internship program (not to be confused with the residency period following their graduation). The internship is carried out in one of 23 possible hospitals in the country, which are allocated interns relative to the size of their patient population (with an advantage to peripheral hospitals), with the smallest allocation being of 4 interns. The internship is perceived as a tax that must be paid: there is (almost) no correlation between the hospital in which internship was performed to the hospital in which residency is performed, interns work nights and weekend to receive below minimum wage, and are not getting any tuition during the internship.\footnote{One justification for this tax is that medical studies are highly subsidized by the government, with a tuition cost of around $2500 \$$ per year. In return, the government uses the interns as cheap labor in hospitals (all hospitals are operated by the government).}

By and large, interns are not assigned to hospitals on the basis of merit, since the government wants to spread the more talented interns everywhere.\footnote{There is an exception to this rule - five interns with a PhD get to choose where they want to be assigned.} Since merit is not a criterion, it makes sense to use some form of lottery to assign the interns. For many years, the lottery which was chosen was \textit{Random Serial Dictatorship} (henceforth RSD).\footnote{Often referred to also as ``Random Priority''. See, for example, \cite{as2003}.}, with some house rules.

In this thesis, we describe the market and the different populations of students, discuss the old lottery, present our new design for the lottery, which was used first in 2014 (and is also in use this year) and present the results of the implementation. We focus this work on the lessons we learned, and on techniques that can be used in designing future markets. More specifically, this thesis deals with
\begin{enumerate}
  \item the ``Do No Harm'' principle, and the constraint that under some metric no student will be worse off in the new match,
  \item the trade-off between efficiency and truthfulness, and
  \item two-body problems in the context of the assignment problem.
\end{enumerate}
In addition, we describe the market, the participants, and share the available data.

\subsection{Background on the internship market}

Each year, two different cohorts of students enter independent lotteries which determine where each student performs her internship. The number of interns allocated to each hospital is determined independently for each cohort.

The first cohort is composed of students who are in their final year of their medical school in Israel. This cohort is the more interesting and diverse one, and contains four populations:
\begin{enumerate}
  \item a handful of students (usually with a PhD) who get to choose their internship. From a market design perspective, they can be treated as reducing the capacity.
  \item Couples who wish to be in the same hospital. Unlike the (American) National Residency Matching Program (NRMP), there is no notion of preferences over pairs of hospitals. This makes sense: in the NRMP one spouse may want to be a radiologist, and the other may want to be a psychiatrist (so they need to rank pairs of programs), whereas in the current match everyone wants to (well, doesn't want to, but has to) be an intern. Around $10\%$ of the interns are involved in a couple, and most couples are not married, but are just friends who want to share an apartment if they are stuck in some village where they do not know anybody. Unlike the NRMP, couples are \textit{guaranteed} to be interns in the same place.
  \item Students who have kids. This population is characterized by not wanting to transfer the family. Hence, they are usually willing to take any hospital which is a driving distance from home, and are less sensitive to hospital's quality.
  \item The rest of the students.
\end{enumerate}

In the last decade, the lottery that was used was Random Serial Dictatorship, with a few small modifications:
\begin{enumerate}
  \item The five students who get to pick, choose where they want to go.
  \item A couple is considered as one student for the matter of choosing the permutation. Alternatively, each couple gets two consecutive spots. If the couple's top choice among hospitals which have vacancies is $A$, but $A$ has just one free spot, they are not allowed to go there and will keep going down their preferences.\footnote{It is not clear what happens if all hospitals have just one vacancy when a couple is chosen. Maybe it just did not happen in recent years. The medical students believed that couples will always stay together, and this was one of their requirements.}
  \item Every year there is a vote about what to do with students with kids (sometimes an advantage is given only to students with two or more kids). The options are to treat them as regular students, to let them take part in the lottery but promise them a seat in their district (the country is divided into five districts for that matter), or to promise them a seat in their district, without participation in the lottery (so they lose the chance to get good hospitals). When the last option is being used, most parents do not declare themselves parents, as they claim that this is not really a benefit.
\end{enumerate}

The second cohort consists of graduates of foreign schools, who want to become Israeli doctors. This only applies to students who studied in certain countries and have no experience, and most of this cohort is composed of Israelis who were not accepted to medicine in Israel, and studied in Jordan or in Europe with the intention of going back to Israel after graduation.

For the foreign cohort, the internship lottery is run by the Ministry of Health (MoH). There are no votes or appeals, and rules are simpler. For the Israeli cohort, the MoH just decides on the capacities, and delegates the rest of the lottery process to a committee of students, elected by the student body. It is common that the committee puts important decisions (such as the parental benefits) to vote by the entire student body.

\section{The redesign process}
We were first approached in 2010 by several medical students who asked for help in redesigning their market. They had two things which worried them:
\begin{enumerate}
  \item giving fair benefits to parents. They felt that parents should be treated differently, but it should not be at the expense of other students, and were not sure how to do this, and
  \item they wanted to improve the efficiency of the system.
\end{enumerate}
We took part in the internship committee of 2011, to understand their demands better. They added a couple of technical requirements:
\begin{enumerate}
  \item {\em Do No Harm}: The redesign process should not hurt any population of students. Everyone should be (weakly) better off in the new design (at least in expectation).
  \item In particular, students should have the option of matching together, just like they did up until then. A couple should get a guarantee that they will be matched together.
\end{enumerate}

The 2011 class was also generous enough to share their data with us. Our initial idea was to use some form of Probabilistic Serial \cite{bm2001}. However, the results were very similar to RSD, and we wanted to improve efficiency. Hence, we decided to use a {\em rank efficient} mechanism, which essentially amounts to maximizing some linear function on the number of students who get their $i$'th rank, for every $i$ (more on that in \autoref{sec:TheNewMechanism}).

To cope with the DNH principle, we ran surveys, in which we gave students two possible vectors of probabilities (for being assigned to each hospital), and asked them which is better. We took the data, and tried to find a simple utility model for it. The model we ended up using is that there are $m$ hospitals, and you have probability $p$ to get to rank $i$, this gives you $p (m - i + 1)^2$ points of happiness, and a profile with more happiness points is better. While one could find a better fit to the model using more parameters, we were happy that there are no constants that we need to compute. We also defined a different happiness function for parents, but it was not used (see below).

Given that we have a definition of happiness, we first compute for every intern what would be her expected happiness under RSD. Then we choose allocation probabilities to maximize total happiness, conditioned on capacities and on the DNH principle -- that every agent is happier under the current algorithm. Given the probability matrix, we decompose it to a convex sum of assignment matrices (more on that in \autoref{two:body}), and choose a matrix at random (according to the weights).

\paragraph{Switching to the new mechanism}
We proposed the new mechanism to the class of 2012, but were rejected. They did not want to have anything to do with computers, and were afraid of bugs and backdoors.

For the 2013 class we approached the MoH. We understood that the MoH is less conservative than the student body, as it has more to gain from new ideas: if the student body is given a new idea who would make it better off for with probability $0.4$ and worse off with probability $0.6$, then they should reject it. However, under these terms it is in the MoH's best interest to try the new idea. If the new idea is a failure than one class pays the price, but if it is a success then future generations will use it as well. In the MoH we found moral support, more data, but no willingness to try the new system. Their main argument was that they do not want to interfere in something which is traditionally run by the students, and are not willing to use the foreign trained students as guinea pigs.

For the 2014 class we went back to the student internship committee. At this stage we had a large quantity of historical data to present about both algorithms, and also had a new idea: we will not run the lottery with a computer, but rather hand them the decomposition of the matrix into a weighted list of assignments, and let them choose the assignment. This way they can compute their own marginal distributions, and verify that there is no backdoor.\footnote{To be more exact, they can compute their marginal distributions to see that they are better off, and the distributions of the members of the committee to see that no one is cheating.} This persuaded the committee to put this for a vote (along with the parental benefits for that year), with the condition that the new algorithm will be used only if it wins an absolute majority. 80\% of the medical graduates who participated in the poll (55\% of the students) voted in favor of the new approach. Following the successful vote, we indeed implemented the algorithm and it was used to assign interns on 2014. In addition, it was used for the foreign cohort that year. A subsequent vote of the following cohort of interns, that took place on January 2015, certified the continued use of our algorithm for 2015, and the foreign trained interns will keep using it until further notice. In order to increase transparency we explained the student body the rules of the new match \cite{bharshm2015}.

We have ran a survey in the class of 2014, presenting students with their marginal distribution under RSD and under the new algorithm, and received supportive feedback. However, we were asked by the student committee to distribute the survey only to committee members, as we ran it after the lottery was over and they wanted to let the dead rest. This turned out to be very unfortunate, as the class of 2015 was mad at us for not handing each student in 2014 his or her marginal distribution. We have agreement from this year's committee to perform the survey with the entire body of students.

In the years 2014 and 2015, the students voted against giving any parental benefits. We think that the main source for this vote is not our algorithm, but a new school of medicine which accepts students who already earned an undergraduate degree (usually in Biology) and want to undergo retraining and perform a career change. The graduates of this school constitute the majority of parents in the market, and therefore it became an issue of ``helping students in the other school'' which attracts less sympathy than ``helping the parents in my class''.

\paragraph{Truthfulness and efficiency}
Our main concern with the new mechanism is that we sacrifice truthfulness (at least to some degree) to achieve efficiency. This is probably a necessity: \cite{liu2013ordinal} showed that in large random market all asymptotically efficient, symmetric, and asymptotically strategy-proof mechanisms are equivalent to Random Serial Dictatorship. While in theory it could be that our market is not large enough, or that there is something special in the valuation profiles we face, it is unlikely that this is the case. Hence, any improvement in the welfare will lead to a non-truthful mechanism, such that some agents will be able to deviate and make a non-negligible gain.

When choosing between truthfulness and efficiency, we need to remember that truthfulness is a means, while efficiency is an end. We see two main arguments for truthful mechanisms:
\begin{itemize}
  \item Truthful mechanisms reduce cognitive burden for the participants.
  \item When bidding is truthful, the designer can evaluate the welfare. When bidding is not truthful, such evaluation is not possible.
\end{itemize}
We think that for this market, one can mitigate the disadvantages of a non-truthful algorithm, and the gain in welfare is big enough to warrant some degree of untruthfulness. Specifically, we believe that bidding today is truthful (as much as anything can be truthful in ranking 25 priorities), and that if in the future a large body of students would bid untruthfully, then discussions in forums and on Facebook revolving around strategic bidding would emerge. We know that in the past, it was a common (and illegal) practice to sell and buy internship positions after the lottery.\footnote{In theory, there should be no trade following a lottery if monetary transfers are not allowed. In practice, internship begins at least ten months (and sometime eighteen months) after the lottery is conducted, so it makes sense that someone who would want a desired hospital in city A would want to move to city B for some personal reason. Moreover, as mentioned, illegal monetary transfers were conducted. The students would approach the MoH, asking to trade places, hiding any financial agreement, and would be granted the permission to trade. The market price of a good internship used to be 2500\$.} Negotiating these deals was done through social media, and it was not too difficult for us to find it. Discussing bidding strategies should happen through the same channels, and be less hidden (since there is nothing wrong about it).

\section{The new mechanism} \label{sec:TheNewMechanism}
We know that after the assignment is done, there is no room for trade. Still, we want to let the students trade to achieve a better allocation. Therefore, instead of trading seats in hospitals, we trade probability shares. Indeed, a student $i$ attending the lottery should not care too much about the internal mechanics $-$ all he should care about is $p_{i,1},\ldots p_{i,23}$ where $p_{i,k}$ is the probability student $i$ gets his $k$'th choice.

Looking at the probability vector each student gets, RSD turns out to be far from optimal. Consider for example the following fictitious lottery, which involves four students and four hospitals, each with capacity 1 [Table \ref{tab:example}].

\begin{table}[htbp]
	\centering
	\caption[Example of students' preferences]{The four students' prioritized preferences over hospitals. All students ranked hospitals A and B as their first or their second choice respectively. Alice and Diane each ranked hospital C as their third choice and then hospital D, while Bob and Charlie each ranked hospital D as their third choice and C as their fourth.\vspace{0.5 cm}}
	\label{tab:example}
		\begin{tabular}{|c|c|c|c|}
			\hline
			\textbf{Alice} & \textbf{Diane} & \textbf{Bob} & \textbf{Charlie} \\ \hline
			A & A & A & A \\ \hline
			B & B & B & B \\ \hline
			C & C & D & D \\ \hline
			D & D & C & C \\
			\hline
		\end{tabular}
\end{table}

We analyze the probability that Alice gets each hospital in this example:
\begin{itemize}
\item With probability $\frac{1}{4}$ she is the first to choose. In this case she chooses $A$, and therefore $\Pr(A) = \frac{1}{4}$.
\item With probability $\frac{1}{4}$ she is the second to choose. In this case $A$ is already taken, and therefore she chooses $B$. Therefore $\Pr(B) = \frac{1}{4}$.
\item With probability $\frac{1}{4}$ she is the third. In this case, $A$ and $B$ are already taken, and Alice chooses $C$. But this is not the only way Alice can get $C$. If Alice is the last to choose, then it is possible that $C$ is still open, and she can take it. Indeed, if Diane goes first and takes $A$, Bob goes second and takes $B$, Charlie is third and takes $D$, then Alice can get $C$ although she is the last to choose. Looking at this more closely, one can see that if Alice is fourth and Diane is not the third, Alice would also get $C$. This means that the total probability that Alice gets $C$ is $\Pr(C) = \frac{5}{12}$.
\item With probability $\frac{1}{12}$, Alice is fourth, and Diane was third. This is the only case in which Alice gets $D$, and therefore     $\Pr(D) = \frac{1}{12}$.
    \end{itemize}
\noindent One can verify that the sum of probabilities is 1, and Alice is always assigned to a hospital.

    A similar argument shows that Bob has probability of $\frac{1}{4}$ to get $A$, probability of $\frac{1}{4}$ to get $B$, probability of $\frac{5}{12}$ to get $D$ and probability of $\frac{1}{12}$ to get $C$.

    Note that in this simple example, Alice and Bob could trade probabilities, and this would benefit both of them. Imagine that Bob could somehow give Alice his probability of being assigned to $C$, and in return she would give him her probability of being assigned to $D$. This would result in a state in which Alice has probability of $\frac{1}{4}$ for $A$, $\frac{1}{4}$ for $B$ and $\frac{1}{2}$ for $C$. Bob would have probability of $\frac{1}{4}$ for $A$, $\frac{1}{4}$ for $B$ and $\frac{1}{2}$ for $D$, which is an improvement for both of them, compared to the RSD probability shares.    Charlie and Diane could trade probabilities among themselves in a similar manner.

    While this already improves the current state of affairs, we can do even better. Suppose that Alice and Bob agree on their first and tenth choices, but Bob's second choice is some hospital $H$, which is also Alice's ninth choice. Also, suppose that Alice has some positive probability $p$ of being assigned to $H$. In this case both students would possibly be happier if Alice ``gives'' Bob  her probability $p$ of getting to $H$, and Bob would ``give'' her $\frac{p}{2}$ probability of getting to his first place, and $\frac{p}{2}$ probability of getting to his tenth place (see Figure~\ref{fig:AliceAndBob}). In this case:

    \begin{enumerate}
      \item Assuming there is no huge gap between the ninth place and the tenth place, Alice should be happy. She ``lost'' probability $p$ of getting to her ninth place, and received probability of $\frac{p}{2}$ to get to her tenth place (similar), and probability of $\frac{p}{2}$ to get to her first place.
      \item Assuming there is no huge gap between the first and second place, Bob should be happy. He lost $\frac{p}{2}$ probability from his first place and $\frac{p}{2}$ probability from his tenth place, to get $p$ probability for his second place.
    \end{enumerate}
		
\begin{figure}[]
  \caption[Trading probabilities example]{Alice gives Bob her probability of being assigned to $H$, and in return she gets half of this probability to her first choice and half of it to her last choice from Bob.}
	\centering
		 \includegraphics[scale=0.7]{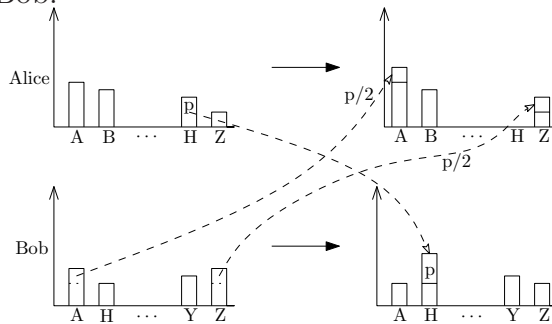}
	\label{fig:AliceAndBob}
\end{figure}

    While clearly this trade is beneficial, it raises a couple of subtle points, which are related:

    \begin{enumerate}
      \item Why should Bob give $\frac{p}{2}$ of his tenth place and $\frac{p}{2}$ of his first place? Why not $\frac{2p}{3}$ of his tenth place, and $\frac{p}{3}$ of his first place, or vice verse, or some different numbers?
      \item The difference between first and second place is usually larger than that between ninth and tenth place.
    \end{enumerate}

    As explained below, the students were asked to fill surveys, to assert the difference between the first and the second place, the second and the third place and so on. Based on the surveys results', more weight was given to the difference between first and second place than to the difference between the ninth and the tenth.

     One question which was not addressed in these simple examples is how to decide which students should trade with whom, and what trades to perform. Of course, the students do not actually trade with each other, but rather a computer program ``virtually'' trades on their behalf. We also use more complex trades, which may involve three or more students at once, if they benefit all the participants.

     Another question is how do we perform the lottery at the end of the process? With RSD, we could just choose an ordering over the students. But in the first example, what lottery gives Alice probability of $\frac{1}{2}$ to get to $C$ and Bob probability $\frac{1}{4}$ to get to $D$? Note that if each student would just select a random hospital according to the probabilities, it is possible that two students would be assigned to the same hospital, so this is not a valid solution.

     In the rest of the section, we formally explain our method, and solve the two questions presented above.

\subsection{Description of the new lottery}
Our method works as follows. First, approximate the probability of each student to be assigned to each hospital using RSD. We do this by running a large number of trials, $N$, and by the law of large numbers the average of all those RSD lotteries will be sufficiently close to the true value. The probability is calculated as

\begin{equation}
\label{eq:prob}
 p_{i,k}=\frac{n_{i,k}}{N},
\end{equation}
where $n_{i,k}$ is the number of RSD lotteries in which Student~$i$ was assigned to the hospital he ranked as his $k$-th choice.

Once we have the approximated probabilities we continue with the second stage of the algorithm which is trading the probabilities among the students. We do this using \textit{Linear Programming}, which is a mathematical optimization method for maximizing a target function subject to several constraints \cite{Kantorovich}.\footnote{Similar approaches to object assignment have been already suggested. See, for example, \cite{featherstone2014} and \cite{hz1979}.} In our interns assignment problem there are two constraints:

\begin{enumerate}
	\item Each hospital has an upper limit for the number of interns that can be assigned to it. This \textit{capacity constraint} is determined by the Ministry of Health.
	
	\item No student is worse off compared to what he would have got under RSD. This \textit{individual rationality constraint} is enforced by defining the \textit{happiness} of each student from his vector of probabilities, and then requiring that for every student individually the happiness can not decrease by the trading stage (intuitively if this would decrease his happiness the student would not trade)
\end {enumerate}

As for the target function, we want to maximize the total satisfaction of the students after trading. The full description of the constraints and the target function appears in Appendix~\ref{app:A}.

After the optimal probabilities have been acquired, we only need to randomize an allocation according to these probabilities. This, however, is not an easy process, as we want the lottery over valid allocations to respect all the interns' probability allocations simultaneously. Fortunately, the Birkhoff-von Neumann theorem provides a solution to this problem \cite{birkhoff1946, vonneumann1953} . In order to apply the theorem for the allocation problem, we represent the probabilities data we have until this stage with a matrix of size $n \times m$, where the rows are the interns, the columns are the hospitals, and cell $(i,j)$ represents Student~$i$'s probability to be assigned to Hospital~$j$. The theorem ensures that any random assignment of the objects in the rows of the matrix to the objects in the columns can be implemented. Furthermore, Birkhoff's proof provides a constructive algorithm for the implementation \cite{pml1986}. Using an extension of this theorem we can create a lottery which respects the improved probabilities gained by trading \cite{bckm2013}.

\section{The new lottery results}
Figure~\ref{fig:cumulativePlace} depicts the number of students who were assigned to one of their top $k$ choices as a function of $k$. The gap between the two curves shown in the figure (the area under the dashed curve) represents the improvement of the new method compared to RSD. While RSD would assign 203 interns to their first choice hospital, 50 to the second and 59 to the third, the new method assigns 216 interns to their first choice, 84 to their second choice and 70 to the third choice.

\begin{figure}[]
  \caption[Comparison between the old and the new methods]{The number of interns as a function of being assigned one of their top choices. The solid line shows the results of RSD, and the dashed line shows the result of the new method.}
	\centering
		 \includegraphics[scale=0.42]{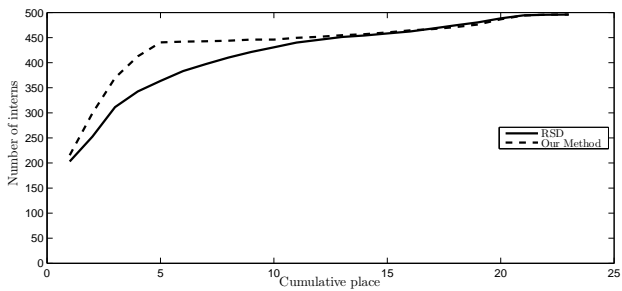}
	\label{fig:cumulativePlace}
\end{figure}

Furthermore, using our data we would like to rate hospitals, according to how high the interns ranked them. Such a rating can be useful, since it gives the hospitals a better picture of their status, and raises a red flag if a specific hospital should improve the way it treats its interns or signals quality to the Ministry of Health. Before we aggregate the data to create such a rating, it is useful to look at the ranking distribution of a specific hospital, and see if it makes sense.

For example, Hadassah Medical Center's ranking distribution is depicted in Figure~\ref{fig:HhadassahRanks}. About $10\%$ of the interns ranked Hadassah as one of their top three choices, which possibly indicates that they are residents of Jerusalem and that location is important to them. The rest of the students ranked Hadassah around the middle of their rank order list, suggesting that interns consider Hadassah to be a good hospital, and that the demand for being an intern there is quite high.
\begin{figure}[]
  \caption[Hadassah hospital's rankings]{The number of interns who ranked \textit{Hadassah} hospital in each place.}
	\centering
		 \includegraphics[scale=0.45]{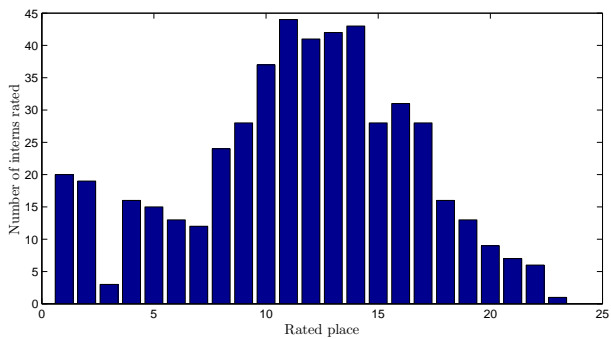}
	\label{fig:HhadassahRanks}
\end{figure}

We now aggregate the rankings of all students, to create a rating across hospitals. We compare two methods of rating the hospitals. The first method is to use the same weights that we used when defining students' happiness, as described in Eq.~\ref{eq:happinessBefore}.\footnote{The rating we get using this method is very similar to the rating one gets using the more familiar Borda Count method.} The second method is the traditional rating of hospitals in this lottery, which is based on the number of interns who ranked a specific hospital as their first choice. Figure~\ref{fig:Ranks} demonstrates that our new rating approach provides very different results from the traditional rating approach, and it is perhaps more advisable to use it as it takes into account the entire rank order lists of all students.

\begin{figure}[]
  \caption[Comparison between two rating approaches]{A comparison between two rating approaches. \textbf{[A]} The number of interns who
	ranked the given hospital as their first choice. \textbf{[B]} The calculated value of the given
	hospital according to the survey-based weights.}
	\centering
	\includegraphics[scale=0.45]{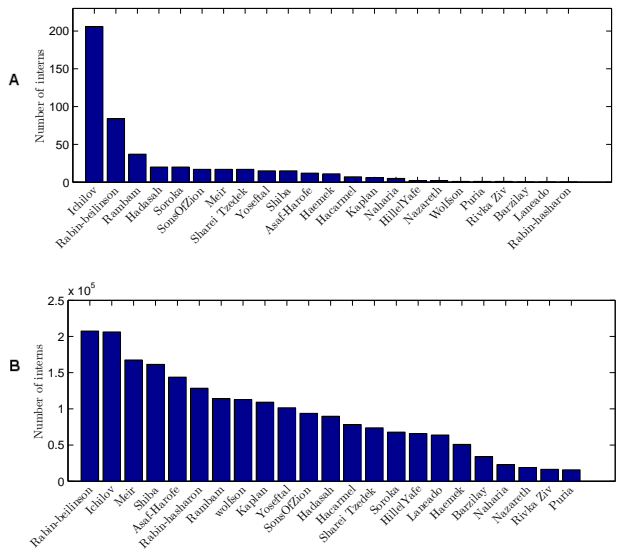}
	\label{fig:Ranks}
\end{figure}

Comparing the two ratings, one can see several differences:
\begin{enumerate}
\item   In the traditional approach, Rabin-Hasharon comes out last, although Rabin-Beilinson came in second. The reason that this happens, is that every student rates Rabin-Belinson above Rabin-Hasharon, so Rabin Hasharon is never first.
    However, the difference is very small - they are both campuses of the Rabin Hospital, and are 10 minutes apart.\footnote{Previously they appeared together in the internship forms under "Rabin"  and a second lottery was performed between the students to see which student goes to which campus} In the new rating, Belinson comes in first (but almost at a tie with Ichilov), and Hasharon comes sixth.
\item The last hospitals in the traditional rating have very low scores, and a single student who changes his vote can change the rating. The new rating is much more robust.
\item The decay in the score makes more sense in the second rating. It is never too sharp, and there is no huge difference between the first two places.
\end{enumerate}

We take the new rating as further evidence that the values at which the algorithm trades probabilities make sense.

\section{Couples in the assignment problem} \label{two:body}

This section will deal with the students internship committee's second condition, and how we were able to comply with it. In the past, any two students were allowed to declare that they are a ``couple'' and be matched together (that is, to the same hospital) by receiving only one joint turn in the Random Serial Dictatorship. The committee required us to give couples the same option under the new algorithm. Similar to several other algorithms for assignment of indivisible objects (e.g., \cite{hz1979,bm2001}), our algorithm ``trades'' probabilities among participants (based on their preferences), and reaches a stochastic matrix that has to be decomposed into a convex combination of valid assignments. The last step uses the famous Birkhoff-von Neumann decomposition process \cite{birkhoff1946,vonneumann1953}. The fact that couples cannot be split imposes (hard) complementarity constraints, which invalidates some of the assignments, and consequently the decomposition process may become harder or even impossible. On the bright side, only stochastic matrices in which both members of the couple get exactly the same probability for each hospital are candidates for decomposition. Is this little advantage enough to make the problem solvable? And if not, what could be done?

In this work we first show that not only is the problem often not solvable when there are couples, but also that it is NP-hard to determine whether a given stochastic matrix with couples can be decomposed (Theorem~\ref{th:np_hard}). This result extends trivially to environments with more general complementarities between groups of participants.

Our second main contribution is in bypassing this impossibility result by providing a polynomial time approximation algorithm that outputs a convex combination of valid assignments that is similar to the target stochastic matrix. We show that the approximation is tight and behaves like $2/\underline{q}$, where $\underline{q}$ is the capacity of the smallest hospital participating in the match. In the last section of the thesis we consider several extensions of the algorithm that can be useful for other applications. We also use anonymized preferences data provided by the MoH to subsample and test our algorithm performance. We find that our algorithm often performs significantly better on actual data (compared to the theoretical bounds we establish).

\subsection{Our Results} \label{subsec:our_results}

\begin{enumerate}
\item It is NP-hard to determine whether a stochastic matrix with couples can be decomposed into a convex combination of deterministic assignments with couples.
\item It is possible to approximate such a decomposition in polynomial time. The approximation is tight and its quality is reciprocal to the smallest hospital' capacity.
\item Subsampling real preferences data from the Israeli Medical Internship Match for which the approximation was developed shows excellent performance of the approximation algorithm.
\end{enumerate}

\subsection{Related Work} \label{subsec:related_work}
In the past two decades, mechanism  design has been applied in ubiquitous contexts and markets. The first and foremost is auctions \cite{ashlagi2009ascending,ashlagi2010position,aumann2014auctioning,chen2010robust,
hassidim2011non,DBLP:conf/wine/HassidimKMN12}, but other examples include division of goods in a fair way
(e.g. cake cutting 
\cite{aumann2013computing,segal2015waste,segal2015envy,segal2014fair}
studying the effect and utilization of networks 
\cite{danna2012upward,hassidim2013network,sina2015adapting}, matching markets
\cite{DBLP:journals/ior/AshlagiBH14,hajaj2015strategy}, and other topics in social choice
\cite{ashlagi2010monotonicity,ajtai2009sorting,azaria2013movie,lutomirskibreaking,
sofer2012negotiation,hassidim2014local,hassidim2014approximate,
cook2013grow,farhi2012quantum}.

The assignment problem with couples is a specific case of multi-unit allocation of indivisible objects without transfers (for example, course allocation), and we indeed intend for our treatment to convey a message regarding the more general case. \cite{budish2011} provides a deterministic algorithm to find an \textit{ex-post} efficient allocation that represent an approximately competitive equilibrium from approximately equal incomes. \cite{bckm2013} suggest a decomposition in the spirit of Birkhoff and von-Neumann, but complementarities are not allowed. A contemporary working paper by \cite{npv2014} is the closest to ours, and it solves the problem of allocating bundles of indivisible objects by assuming (like we do) that bundles are limited in size, and that the capacity constraints are ``soft''. While we insist on capacities being met, allowing deviations as those of Nguyen et al. and then correcting them in a smart way will give a result similar to our approximation result (Theorem~\ref{th:approximationAlgorithm}, see also Section~\ref{sec:extensions}). \cite{an2014} suggest a different kind of approximation, by ignoring the couples constraint with a small probability.

The motivation for decomposing stochastic assignment matrices comes from mechanisms that efficiently allocate probabilities in the interim stage. Notable contributions for the case of single-item assignments are \cite{hz1979} who proposes a competitive equilibrium from equal incomes approach, \cite{bm2001} who suggest the Probabilistic Serial mechanism which is ordinally efficient (i.e., interim efficient given ordinal preferences and not cardinal preferences), and \cite{featherstone2014} who studies rank-efficient mechanisms (similar in spirit to the one implemented by us for the Israeli Medical Internship Match). We note that several authors have dealt with the case of object assignment when objects are the endowment of some agents, in which case the top-trading cycles (or a variation thereof) is often the mechanism of choice (see, e.g., \cite{ss1974,as2003}).

Finally the effect of couples in matching problems has also been studied in the context of two-sided matching. See, for example, the works by \cite{kpr2013} and \cite{abh2014}. A treatment more related to our current approach is provided by \cite{nv2014}.

\subsection{Model and Notation} \label{sec:model}
An assignment problem with couples is a tuple $\left(S, C, H, M \right)$.  $S$ is a finite set of single interns, $C$ a finite set of couples of interns, with each $c \in C$ being a set of two interns $c = \{c_1,c_2\}$. We denote by $I = S \cup \left(\bigcup_{c \in C} c\right)$ the set of all interns. $H$ is a finite set of hospitals. $M \in [0,1]^{I \times H}$ is \textit{the target matrix} and it satisfies
\begin{enumerate}
\item $\forall i \in I: \sum_{h \in H} M_{i,h} = 1$, and
\item $\forall (c_1,c_2) \in C, h \in H: M_{c_1,h} = M_{c_2,h}$.
\end{enumerate}
We let $q_h = \sum_{i \in I} M_{i,h}$ be the capacity of hospital $h$, and $\underline{q} = \min_h q_h$ be the capacity of the smallest hospital. Let $\mathcal{P}$ be the domain of all assignment problems with couples. We specify two special sub-domains of $\mathcal{P}$: $\mathcal{P}^{C = \emptyset}$ is the set of problems without couples (i.e., $C = \emptyset$), and $\mathcal{P}^{S \geq C}$ is the set of problems in which singles receive more weight than couples in \textit{every} hospital (i.e., $\forall h \in H: \sum_{s \in S} M_{s,h} \geq 2 \sum_{c \in C} M_{c_1,h}$).

A matrix $M' \in [0,1]^{I \times H}$ that satisfies conditions (1) and (2), and $\forall h \in H: \sum_{i \in I} M'_{i,h} = q_h$ is called \textit{a stochastic assignment matrix (with respect to $P$)}. If all the elements of a stochastic assignment matrix (with respect to $P$) are in $\{0,1\}$, then the matrix is called \textit{a deterministic assignment matrix (with respect to $P$)}. A stochastic assignment matrix $M'$ is \textit{decomposable} if it can be represented as a convex combination of deterministic assignment matrices, i.e., $\left\{\left(\lambda^k, M^k\right)\right\}_{k=1}^K$ such that for each $k$, $\lambda^k \in (0,1)$ and $M^k$ is a deterministic assignment matrix, $\sum_{k=1}^K \lambda^k = 1$, and $M' = \sum_{k=1}^K \lambda^k M^k$. A straightforward extension of the Birkhoff-von Neumann theorem shows that on the sub-domain $\mathcal{P}^{C = \emptyset}$ all target matrices are decomposable. In Section~\ref{sec:np_hard} we show that on the general domain, $\mathcal{P}$, it is NP-hard to verify whether a target matrix is decomposable.

We say that a convex combination of deterministic assignment matrices $\left\{\left(\lambda^k, M^k\right)\right\}_{k=1}^K$ \textit{$\epsilon$-approximates} a target matrix $M$ if
\begin{align*}
\left|\left|| M - \sum_{k=1}^K \lambda^k M^k|\cdot\vec{\mathbf{1}}_{[I]}\right|\right|_{\infty} < \epsilon,
\end{align*}
where $\vec{\mathbf{1}}_{[I]}$ is a vector of ones of size $I$, so we approximate the maximum value over the $L_1$ norm of each row.

A decomposition algorithm $\mathcal{A}$ takes a problem $P \in \mathcal{P}$ and outputs a convex combination of deterministic assignment matrices (with respect to $P$). We say that $\mathcal{A}$ provides an \textit{$f$-approximation on domain $\mathcal{P}'$} if for every $P \in \mathcal{P}'$, $\mathcal{A}(P)$ $f(P)$-approximates $M$. Our aim in Section~\ref{sec:approximation} is to provide a lower bound for the optimal approximation on $\mathcal{P}$, and an upper bound for the optimal approximation on $\mathcal{P}^{S \geq C}$.

\subsection{NP-hardness result} \label{sec:np_hard}
\begin{theorem} \label{th:np_hard}
On the domain $\mathcal{P}$, determining whether the target matrix is decomposable is in NPC.
\end{theorem}

\begin{proof}
We will reduce 3-edge coloring of cubic graphs to the interns-hospitals
assignment problem, (\cite{holyer1981}, proved that 3-edge coloring of
cubic graphs is in NPC).

Let $G=(V,E)$ be a cubic graph with $\left|V\right|=n$ vertices and $m=3n/2$ edges.
For each edge $e$ let us have $3$ hospitals $A(e), B(e), C(e),$ each of capacity
$2$. For each vertex $v$ and each color $i \in \{1,2,3\}$  let us have a
hospital $(v,i)$ with capacity $1$. Thus altogether we have $3m$ hospitals
of capacity $2$ and $3n$ hospitals of capacity $1$, total hospital capacity
is $3n+6m=3n+6 \cdot{3n/2}=12n$.

Now for each edge $e=(u,v)$ we have one couple that wants to be
either in $A(e)$ or in $B(e)$ or in $C(e)$, each with probability $1/3$
and we also have $6$ single interns as follows:

\begin{itemize}
\item First single wants either $A(e)$ with probability $2/3$ or $(u,1)$ with probability $1/3$.
\item Second wants either $A(e)$ with probability $2/3$ or $(v,1)$ with probability $1/3$.
\item Third wants either $B(e)$ with probability $2/3$ or $(u,2)$ with probability $1/3$.
\item Forth wants either $B(e)$ with probability $2/3$ or $(v,2)$ with probability $1/3$.
\item Fifth wants either $C(e)$ with probability $2/3$ or $(u,3)$ with probability $1/3$.
\item Sixth wants either $C(e)$ with probability $2/3$ or $(v,3)$ with probability $1/3$.
\end{itemize}

Note that altogether we have $(2+6)m=8 \cdot{3n/2}=12n$ interns. Hence, any
assignment of all of them  uses all capacity of all hospitals.

Suppose first that $G$ is class $1$ (that is, it is 3-edge colorable). Fix a
proper 3-edge coloring by colors $1,2,3$. Each such coloring corresponds
to a full assignment of all interns as follows:

If edge $e=(u,v)$ is colored $1$ then the couple is in $A(e)$, First single
intern is in $(u,1)$, second in $(v,1)$, third and forth in $B(e)$, fifth and
sixth in $C(e)$. If $e$ is colored $2$ and $3$ we proceed symmetrically in the
obvious way. It's clear that this is a full assignment.

Now give this assignment weight $1/3$, and give weight $1/3$ to each of the
two other assignments obtained from it by cyclically shifting the colors
of all edges.  This gives a decomposition of our matrix of probabilities.

Conversely, if there is a decomposition, then any single assignment
in its support must be a full assignment (assigning all interns and
saturating all hospitals). But this means that for each edge $e$, among
the $8$ interns corresponding to this edge, $6$ including the couple are
assigned to the hospitals $A(e),B(e),C(e)$ and there is a unique $i$
$\in\{1,2,3\}$ so that two of the single interns are assigned to $(u,i)$
and $(v,i)$ and thus $G$ is 3-edge colorable, as needed.
\end{proof}

\subsection{Approximation} \label{sec:approximation}

\subsubsection{Lower bound}

\begin{theorem}
\label{th:lowerBound}
If $\mathcal{A}$ provides an $f$-approximation on $\mathcal{P}^{S \geq C}$, then for any $n \in \mathbb{N}$ there exists $P \in \mathcal{P}^{S \geq C}$ such that $\left|I\right| \geq n$ and $f(P) \geq \frac{2}{\underline{q} + 2}$.
\end{theorem}

\begin{proof}
Given $n \in \mathbb{N}$, let $q = 4 \lceil \frac{n}{8} \rceil$. Let $S = \left\{s_1,\dots,s_{\frac{1}{2}q+1},s'_1,\dots,s'_{\frac{1}{2}q+1} \right\}$, $C = \left\{ c_*, c_1,\dots,c_{\frac{1}{4}q-1},c'_1,\dots,c'_{\frac{1}{4}q-1} \right\}$, and $H = \left\{h, h'\right\}$. Consider the target matrix $M \in \mathcal{P}^{S \geq C}$ described in Table~\ref{tab:example_singles}. Under this stochastic assignment matrix each single intern in $\left\{s_1,\dots,s_{\frac{1}{2}q+1} \right\}$ and each couple in $\left\{c_1, \dots, c_{\frac{1}{4}q-1} \right\}$ gets a probability $1$ to be assigned to $h$, each single intern in $\left\{s'_1,\dots,s'_{\frac{1}{2}q+1} \right\}$ and each couple in $\left\{c'_1,\dots, c'_{\frac{1}{4}q-1} \right\}$ gets a probability of $1$ to be assigned to $h'$, and the couple $c_*$ gets an equal probability to be assigned to either $h$ or $h'$. Note that the capacity of both hospitals is $q$.

\begin{table}[htbp]
	\caption{Target matrix for proof of Theorem~\ref{th:lowerBound}}
	
	\centering
		\begin{tabu}{|c|c|c|}
			\hline
			\textbf{} & \textbf{h} & \textbf{h'} \\ \hline
			
			$s_1$ & 1 & 0 \\ \hline
			\vdots & 1 & 0 \\ \hline
			$s_{\frac{1}{2}q + 1}$ & 1 & 0 \\ \tabucline[1pt]{-}
			
			$s'_1$ & 0 & 1 \\ \hline
			\vdots & 0 & 1 \\ \hline
			$s'_{\frac{1}{2}q + 1}$ & 0 & 1 \\ \tabucline[1pt]{-}

			$c_{*,1}$ & 0.5 & 0.5 \\ \tabucline[1pt]{-}
			$c_{*,2}$ & 0.5 & 0.5 \\ \tabucline[1pt]{-}
			
			$c_{1,1}$ & 1 & 0 \\ \hline
			$c_{1,2}$ & 1 & 0 \\ \hline			
			\vdots & 1 & 0 \\ \hline
			$c_{\frac{1}{4}q-1,1}$ & 1 & 0 \\ \hline
			$c_{\frac{1}{4}q-1,2}$ & 1 & 0 \\ \tabucline[1pt]{-}
			
			$c'_{1,1}$ & 0 & 1 \\ \hline
			$c'_{1,2}$ & 0 & 1 \\ \hline			
			\vdots & 0 & 1 \\ \hline
			$c'_{\frac{1}{4}q-1,1}$ & 0 & 1 \\ \hline
			$c'_{\frac{1}{4}q-1,2}$ & 0 & 1 \\ \hline
		\end{tabu}
		\label{tab:example_singles}
		
\end{table}

In any convex combination that approximates $M$, one of the hospitals will be assigned at least $\frac{1}{4}q$ of the couples with a weight of at least $\frac{1}{2}$. This means that with probability $\frac{1}{2}$ the left over capacity for single interns is at most $\frac{1}{2}q$, where as the single interns require $\frac{1}{2}q+1$. The optimal way to minimize the deviation from the singles probability would be to divide the deviation equally among all the relevant single interns. This means that each individual intern in the relevant hospital will get at most $\frac{1}{2} \cdot 1 + \frac{1}{2} \cdot \frac{\frac{1}{2}q}{\frac{1}{2}q+1}$ instead of getting $1$, and to complete the individual's total probability to $1$, the single will get $1-(\frac{1}{2} \cdot 1 + \frac{1}{2} \cdot \frac{\frac{1}{2}q}{\frac{1}{2}q+1})$ to the other hospital (instead of $0$). Therefore the approximation cannot be better than:
\[
2\left (1 - \frac{1}{2}\cdot 1 - \frac{1}{2} \cdot \frac{\frac{1}{2}q}{\frac{1}{2}q+1}\right)= \frac{2}{q+2} \left(= \frac{2}{\underline{q}+2} \right)
\]
\end{proof}

One may argue that the example provided in the proof for Theorem~\ref{th:lowerBound} seems quite pathological. Indeed, the target matrix allocates many interns either $h$ or $h'$ with probability $1$. There could be scenarios in which it would be very reasonable to consider only target matrices in which each intern's probability of reaching any particular hospital is bounded by some expression related to the number of hospitals (see also the concluding discussion). Nevertheless, in our main application the chosen algorithm involves trading probability between interns, and so it often outputs target matrices that are close to being pathological in exactly the same sense. This motivates our approximation metric and the consideration of extreme cases. More importantly, a similar bound can be achieved even if probabilities are restricted to being small, although the problems are not in the domain $\mathcal{P}^{S \geq C}$ (see Appendix~\ref{apndx:small_probs}). Finding a lower bound in the domain $\mathcal{P}^{S \geq C}$ when probabilities are constrained to be ``small'' remains an open problem.

\subsubsection{Suggested algorithm for upper bound}
We present an approximation algorithm for decomposing any matrix $M \in \mathcal{P}$. The algorithm can be roughly divided into two main stages: the first stage deals only with couples and assigns them to hospitals according to their probability in $M$, and the second stage 'fixes' the probabilities of the single interns in $M$ to match the assignment of the couples from the first stage. In the second stage the over-demanded capacity taken by the couples in any specific hospital is deducted from singles' demand according to a division which preserves the weight of each single intern in the singles' demand.

The algorithm starts by ``enlarging'' each hospital capacity to the sum of the couples probabilities to be assigned to $h$ divided by two and rounded up to the nearest integer. The couples now behave like singles, and for each hospital we add a new single to take the ``extra'' capacity (the new single's demand is completed by a dummy hospital, $h_\emptyset$). Doing so for all hospitals gives us a new stochastic assignment matrix containing only ``singles'' (with each single representing an original couple), and it can be decomposed using the Birkhoff-von Neumann theorem, to get a convex combination of deterministic assignment matrices that represent the allocation of the couples. In some of these assignment the couples take slightly more capacity than their expected share (but only up to the nearest multiple of two).

In the second stage of the algorithm, we look at the residuals of the first-stage assignment. For each couples assignment and for each hospital we check whether the couples exceeded their expected share in that hospital, and if so, we let all the singles reduce their demand by their respective weights. The missing demand is directed at a dummy hospital, $h_\emptyset$. We again have a stochastic assignment matrix containing only singles that can be decomposed using the Birkhoff-von Neumann theorem.\footnote{Strictly speaking, the demand needs to be rounded up to the nearest integer, and this can be done again by adding dummy interns for each hospital, and putting all the leftover demand in the dummy hospital.} The output of the decomposition assigns some of the singles to the dummy hospital, so we move them to vacant positions arbitrarily. The algorithm then combines back the first-stage and the second-stage assignments to get a valid deterministic assignment matrix with respect to the original problem.

The full algorithm's pseudo-code can be found in Algorithm~\ref{alg:approximation}.

{
\begin{algorithm}
\caption{Approximation}\label{alg:approximation}
Input: A stochastic assignment matrix $M^{I \times H}\in \mathcal{P}$
\begin{algorithmic}[1]
\State $S' = C \cup H$, $H' = H \cup \{h_\emptyset\}$
\State Create a new matrix $M' \in [0,1]^{|S'| \times |H'|}$ (initialize with zeros)
\ForAll {$h \in H$}
	\ForAll {$c \in C$}
		\State $M'_{c,h} = M_{c_1,h}$
	\EndFor
	\State $M'_{h,h} = \lceil \sum_{c \in C} M'_{c,h} \rceil - \sum_{c \in C} M'_{c,h}$
	\State $M'_{h,h_\emptyset} = 1 - M'_{h,h}$
\EndFor
\Comment We get that $(S',\emptyset,H',M') = P' \in \mathcal{P}^{C = \emptyset}$

\State Decompose $M'$ into a convex combination of deterministic assignment matrices $\left\{\left(\lambda^k, M^k\right)\right\}_{k=1}^K$ (with respect to $P'$)
\State Create an empty set of allocations $\psi=\{\}$
\For{$k=1$ to $K$}
	\State $\tilde{H} = H \cup \{h_\emptyset\}$
	\State Create a new matrix $\tilde{M} \in [0,1]^{|S| \times |\tilde{H}|}$ (initialize with zeros)
	\ForAll {$s \in S$}
		\ForAll {$h \in H$}
			\If {$\sum_{c \in C} M^k_{c,h} > \sum_{c \in C} M'_{c,h}$}
			\Comment If couples exceeded their quota
				\State $\tilde{M}_{s,h} = M_{s,h} - \frac{M_{s,h}}{\sum_{s' \in S} M_{s',h}} \cdot 2 \left(\sum_{c \in C} M^k_{c,h} - \sum_{c \in C} M'_{c,h} \right)$
			\Else
				\State $\tilde{M}_{s,h} = M_{s,h}$
			\EndIf
		\EndFor
		
		\State $\tilde{M}_{s,h_\emptyset} = \sum_{h \in H} M_{s,h} - \sum_{h \in H} \tilde{M}_{s,h}$
	\EndFor
	\Comment We get that $(S, \emptyset, \tilde{H}, \tilde{M}) = \tilde{P} \in \mathcal{P}^{C = \emptyset}$

	\State Decompose $\tilde{M}$ into a convex combination of deterministic assignment matrices $\left\{\left(\hat{\lambda}^l, \hat{M}^l\right)\right\}_{l=1}^{L}$ (with respect to $\tilde{P}$)
	
	\For {$l = 1$ to $L$}
		\State Stitch $M^k$ and $\hat{M}^l$ into a valid deterministic assignment matrix $M^{k,l}$ for $P$. Couples get the hospitals they are assigned under $M^k$, singles not assigned to $h_\emptyset$ get the hospitals they are assigned under $\hat{M}^l$, and singles assigned to $h_\emptyset$ get the rest of the vacant positions (arbitrarily)
		\State Add $\left(\lambda^k \cdot \hat{\lambda}^l, M^{k,l} \right)$ to $\psi$
	\EndFor
\EndFor
\State Output $\psi$
\end{algorithmic}
\end{algorithm}
}

\begin{theorem}
\label{th:approximationAlgorithm}
Algorithm~\ref{alg:approximation} provides an $\bar{f}$-approximation on $\mathcal{P}^{S \geq C}$, where $\bar{f}(P) = 2/\underline{q}$.
\end{theorem}

\begin{lemma}
\label{le:polynomialTime}
Algorithm~\ref{alg:approximation} runs in polynomial time.
\end{lemma}
\begin{proof}
The size of the convex combination $\left\{\left(\lambda^k, M^k\right)\right\}_{k=1}^K$ related to the couples is polynomial, also we use Birkhoff's proof which provides a constructive algorithm for the implementation of the decomposition \cite{pml1986} in polynomial time, so the whole algorithm is polynomial.
\end{proof}

\begin{proof}[Proof of Theorem \ref{th:approximationAlgorithm}]
In the first part of the algorithm (lines $1$ through $10$) the couples are allocated exactly their demanded capacity. This is done by taking only the couples, treating them as single agents, enlarging the capacities related to the couples so that they would be an integer, and add a dummy hospital to take all the left-over probabilities of the dummy agents. Then it is possible to use the Birkhoff-von Neumann decomposition.

We note that since each hospital has only one related dummy agent, then the capacity taken by the couples is always either the full (rounded up) capacity of the hospital according to $M'$, or one less than that. If it is the first case, the singles can be allocated their entire demand for this hospital in the next step. However, if the couples exceeded their fractional capacity, we ``fix'' the demand of the single interns, and each intern $i$ loses $2\left( \lceil q'_h \rceil - q'_h \right)$ times his weight in the sum of probabilities of singles, where $q'_h = \sum_{c \in C} M_{c_1,h}$.

Let $y_h$ be the total weight given the matrices of the second kind (for a certain hospital $h$). Then the following holds:
\begin{equation}
\label{eq:secondCaseChance}
(1-y) \cdot \lfloor q' \rfloor + y \cdot \lceil q'_h \rceil = q'_h
\end{equation}
If we denote by $x_h = q'_h - \lfloor q'_h \rfloor$ the fractional part of the couples' total probability we get $y_h = x_h / 2$. Now each time the second case materializes, the singles lose $2 - x_h$. The total loss is given by $(2 - x_h) \cdot \frac{x_h}{2}$, which is maximized at $x_h = 1$. Since we assumed that each hospital contains more singles interns than interns who are part of a couple, each single $s$ losses at most $M_{s,h} \cdot \frac{1}{2} / \frac{q_h}{2} = M_{s,h} / q_h$.

Across all hospitals, each single intern $s$ losses at this point at most $\sum_{h \in H} M_{s,h} / q_h \leq (\sum_{h \in H} M_{s,h}) / \underline{q} = 1/\underline{q}$. However, in the last phase of the algorithm (line 27), singles receive back their ``lost probability'' in arbitrary hospitals, and so each single can deviate up to $1/\underline{q}$ further from her endowment, resulting in a total of at most $2 / \underline{q}$.
\end{proof}

\section{Experimental results} \label{sec:experimental}

In this section we characterize the interns' preferences and provide simulation results using the Israeli Medical Internship Match data provided to us by the Israeli MoH. The data contained the preferences of interns starting from 1995. However, due to significant changes in the hospitals and the internship conditions in them, some of the earlier data cannot be treated as coming from a similar distribution to the more recent data. We therefore focused only on preferences from the years 2010-2014. During those years there were 23 possible hospitals for each intern to rate, and interns \textit{must} rank all hospitals. The number of interns varied from one year to another, but it is always around 500, and the number of internship positions is always equal to the number of interns. Preferences of foreign trained doctors are also available to us, but only until the year 2012.

\subsection{Description of interns' preferences}

Looking at the interns' preferences, there are several interesting points to be made.  The first is that preferences are geographically driven. For example, ranking the hospitals according to the number of students who listed this hospital as their top priority (this is the ranking published by the MoH), we see that in the top five hospitals there are two hospitals from the center, one from the north, one from the south, and one from the Jerusalem area (these are four of the districts in Israel). However, for any given student, such a geographically dispersed ranking is very unlikely. Indeed, except for Ichilov hospital (a hospital in the center of Israel which represents the top option for most of the population), it is common that all the top choices of a student come from the same geographic area.\footnote{It is known that one can always swap Ichilov hospital later for something else. Hence some students pick it first, because they want the flexibility -- the lottery is conducted between ten and eighteen months before the internship starts.} To measure this effect, Figure~\ref{fig:k_same_area} shows the percentage of interns whose top $k$ choices all came from the same area, for $k = 2$ to $10$ (this graph ignores Ichilov hospital). Note that if choices were random, one would expect an exponential decay, where here we see a linear decay.

\begin{figure}[htp]
  \caption[The Israeli and foreign doctors' geographical orientation]{Geographical orientation at the top of the list}
	\centering
		 \includegraphics[scale=0.4]{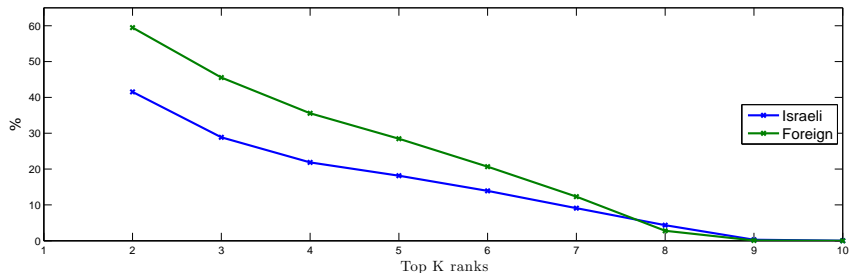}
	\label{fig:k_same_area}
\end{figure}

Another interesting point is the differences between the Israeli interns (interns who studied medicine in Israel) and the foreign trained doctors. Foreign trained graduates have even stronger geographic preferences than Israeli ones. This helps explain another difference between these two populations' preferences: their performance in the old RSD mechanism. Under that mechanism the average rank of hospital that the Israeli interns received was 4.594, whereas the foreign trained doctors (who are matched separately) received an average rank of 2.549. Similarly, under the new mechanism the average rank of Israeli interns was improved to 3.686, and the foreign trained doctors got 2.042. Our hypothesis was that the foreign trained doctors' preferences therefore exhibit more heterogeneity in some sense.\footnote{The common conspiracy theory is that the MoH gives them more seats at the preferable hospitals. We have verified the numbers and the conspiracy theory is incorrect.}

To verify that foreign trained doctors' preferences are more heterogeneous, we grouped preferences according to the top three hospitals, sorted by frequency, and plotted the percentage of students who chose the most common triplet, one of the two most common triplets, and so on. For example, if the most common triplet for the foreign trained doctors is $A,B,C$, and the second most common is $B,C,A$, we are interested in the fraction of students whose top three choices are $A,B,C$, and at the fraction of students whose top three choices are either $A,B,C$ or $B,C,A$.
The results are presented in Figure~\ref{fig:top_triplets}. The right panel shows that the cumulative distribution of local interns is much higher than the cumulative distribution of the foreign trained doctors. The top ten triplets cover almost 50\% of the Israeli interns, but only about 20\% of the foreign trained doctors. Amazingly, even the density of each of the top 10 triplets (which are not necessarily the same among the two populations) is higher for the local interns (as depicted by the left panel). This indicates that their preferences are indeed much more homogeneous, at least at the top of the distribution. We note that Figure~\ref{fig:k_same_area} suggests that foreign trained doctors do exhibit a higher tendency to group hospitals by area. The only reasonable conclusion is that while foreign trained doctors care a lot about geography, there are less prone to select necessarily the center of Israel as their preferred location, or that they tend to have more heterogeneity in ranking hospitals within each area.

\begin{figure}[htp]
  \caption[Distribution of top triplets choices]{Distribution of top triplets (left panel: density, right panel: cumulative)}
	\centering
		 \includegraphics[scale=0.4]{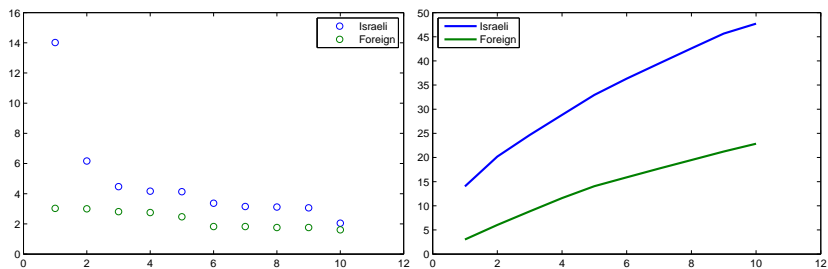}
	\label{fig:top_triplets}
\end{figure}

We also tried to study the differences between the preferences of singles and couples, but did not reach any significant conclusions. Furthermore, while we did not perform a thorough analysis, it is likely to assume that preferences depend also on the medical school in which an intern completed her medical studies, for both personal reasons (already lives in the same city) and professional reasons (knows the hospital better from her time at medical school). This by itself might not have an impact on designing a mechanism, but it may be of interest for understanding better the distribution of preferences and what creates heterogeneity.

\subsection{Simulation results for the algorithm}
In the simulations we used a subsampling method, i.e., we drew preference profiles from the union of all data points, where we distinguish between preferences that were submitted by single interns and preferences that were submitted by couples. This allowed us to create multiple ``possible markets'', each with $|I|=496$ (which was the actual number of interns in 2014) and with $24$ couples. For each point in the figure we used 5,000 different market draws.

To create Figures~\ref{fig:maxL1} and \ref{fig:avgL1}, we drew markets (in the sense explained above), then used our assignment mechanism to generate the matrix $M$, to which Algorithm~\ref{alg:approximation} was applied. While Theorem~\ref{th:approximationAlgorithm} only ensures approximation quality of $2/\underline{q}$, we can clearly see in Figure~\ref{fig:maxL1} that although in our case $\underline{q} = 4$, the average performance of our algorithm is far better than $\frac{1}{2}$ (the red vertical line). One reason for the difference between the theoretical prediction and the data is that our analysis assumes that the percentage of singles' weights is equal across all hospitals (see also the concluding discussion). The second and much more important reason is that even after applying the probability allocation mechanism, pathological examples similar to the one presented in Theorem~\ref{th:lowerBound} are relatively rare. A third reason is that the arbitrary way in which singles were allocated at the end of the algorithm may in fact improve the performance.\footnote{One may consider replacing the arbitrary assignment with a more sophisticated method to slightly improve the results. We did not pursue this direction any further.} We note that the spike at the left side of the distribution occurs exactly because of a small hospital to which couples are likely to be allocated because of their (non-random) preferences (indeed, this spike disappears completely in Figure~\ref{fig:maxRandomL1} which uses random allocation of probabilities).

\begin{figure}[htp]
  \caption[The distribution of the maximum value over the $L_1$ norm]{The histogram shows the distribution of the maximum value over the $L_1$ norm of each intern in the absolute distance between the original assignment matrix and the approximated one. The red dashed vertical line represents the theoretical upper bound.}
	\centering
		 \includegraphics[scale=0.4]{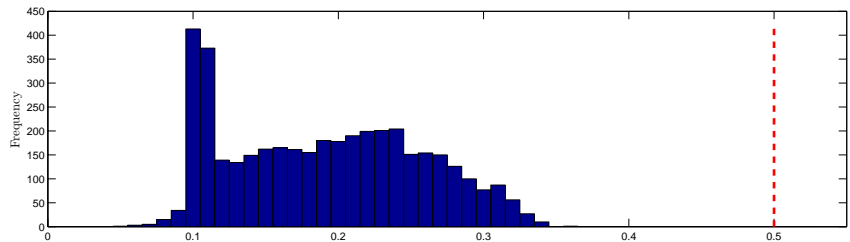}
	\label{fig:maxL1}
\end{figure}

For completeness, Figure~\ref{fig:avgL1} depicts the distribution of $\frac{| M - \sum_{k=1}^K \lambda^k M^k|\cdot\vec{\mathbf{1}}_{[I]}}{|I|}$ , which is the average $L_1$ norm. This metric was not analyzed theoretically, but it is of much interest to the social planner. It roughly means that while the intern who got her probability vector changed the most suffered an average of at about $15\%$ change, the average intern's probability vector was only changed by less than $2\%$.

\begin{figure}[htp]
  \caption[The distribution of the average value over the $L_1$ norm]{The histogram shows the distribution of the average value over the $L_1$ norm of each intern in the absolute distance between the original assignment matrix and the approximated one.}
	\centering
		 \includegraphics[scale=0.4]{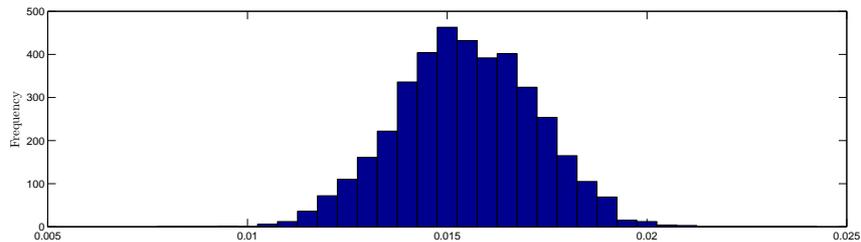}
	\label{fig:avgL1}
\end{figure}

In order to separate the effect of the LP mechanism used to generate probability matrices, we also produce similar figures for matrices produced randomly using the iterative algorithm of \cite{sk1967} (but only select those within our sub-domain of interest, $\mathcal{P}^{S \geq C}$), with the same market size as used for Figures~\ref{fig:maxL1} and \ref{fig:avgL1}. Figure~\ref{fig:maxRandomL1} indeed reveals that using random matrices reduces the maximal impact by a factor of about 5. However, it is evident from Figure~\ref{fig:avgRandomL1} that the average intern's probability vector is now more susceptible to changes following the application of our approximation algorithm.

\begin{figure}[htp]
  \caption[Maximum value over the $L_1$ norm, for random values]{The histogram shows the distribution of the maximum value over the $L_1$ norm of each intern in the absolute distance between the original assignment matrix and the approximated one, for randomly generated values.}
	\centering
		 \includegraphics[scale=0.4]{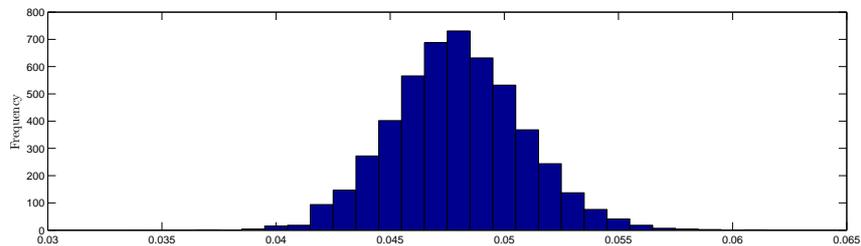}
	\label{fig:maxRandomL1}
\end{figure}

\begin{figure}[htp]
  \caption[Average value over the $L_1$ norm, for random values]{The histogram shows the distribution of the average value over the $L_1$ norm of each intern in the absolute distance between the original assignment matrix and the approximated one, for randomly generated values.}
	\centering
		 \includegraphics[scale=0.4]{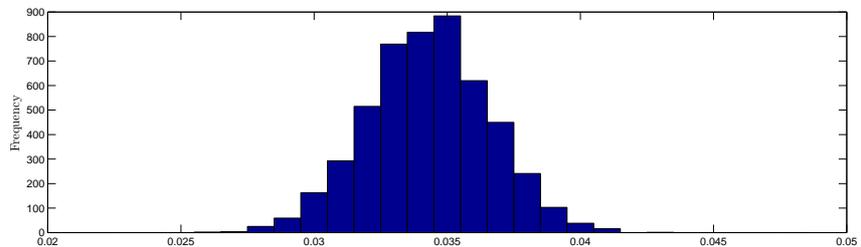}
	\label{fig:avgRandomL1}
\end{figure}

Finally, we also tried to see what would have happened to the performance of our algorithm if couples' preferences were more capacity-driven compared to their true distribution. For this test, we used the same distribution for singles' preferences, but each couples' preference was created by drawing hospitals one by one, where every time the draw is from those hospitals not drawn yet, and with the capacity of the hospital being the weight in the random draw. This made couples like big hospitals better than they like small hospitals, and thus the minimal weight of singles in all hospitals went up, and our algorithm's performance was improved (see Figures~\ref{fig:maxCapacityL1} and \ref{fig:avgCapacityL1}).

\begin{figure}[htp]
  \caption[Maximum value over the $L_1$ norm, for capacity-driven couples]{The histogram shows the distribution of the maximum value over the $L_1$ norm of each intern in the absolute distance between the original assignment matrix and the approximated one, for capacity-driven couples.}
	\centering
		 \includegraphics[scale=0.4]{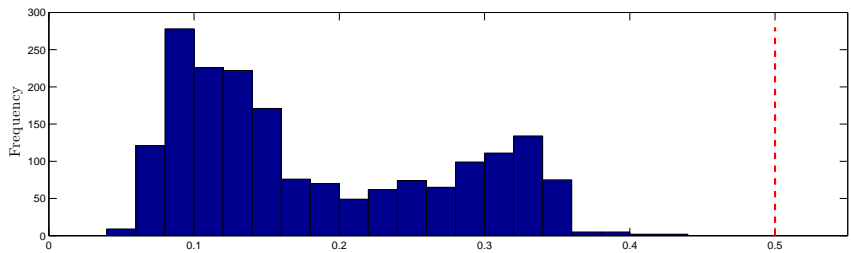}
	\label{fig:maxCapacityL1}
\end{figure}

\begin{figure}[htp]
  \caption[Average value over the $L_1$ norm, for capacity-driven couples]{The histogram shows the distribution of the average value over the $L_1$ norm of each intern in the absolute distance between the original assignment matrix and the approximated one, for capacity-driven couples.}
	\centering
		 \includegraphics[scale=0.4]{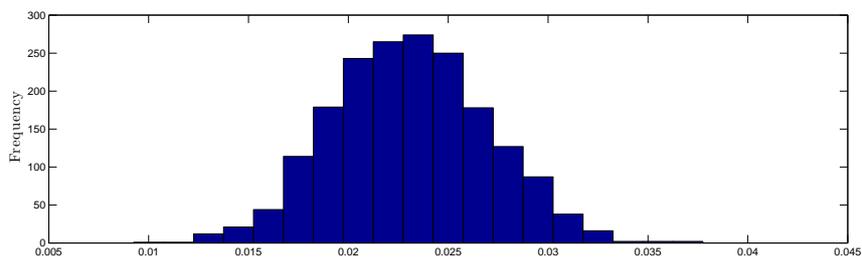}
	\label{fig:avgCapacityL1}
\end{figure}

Figure~\ref{fig:maxCapacityL1} presents a funny double-peaked distribution. The reason for this is not lack of experiments, but rather whether or not a couple has a chance of getting to a hospital with capacity $4$. Since the LP is an affine maximizer, if a couple has a chance to get to this small hospital, then this probability will be non negligible (in the order of at least $0.3$ or so). If the couple has a chance to get there then rounding the probabilities of the intern there will give the intern with the highest discrepancy in the decomposition. One the other hand, if no couple ever gets to that hospital the decomposition there is always perfect, and we need to worry about the second smallest hospital (which already has twice the number of interns than the smallest hosptial).

\section{Extensions} \label{sec:extensions}

\subsection{Lower percentage of singles in each hospital}
For the sake of simplicity we focused on the domain $\mathcal{P}^{S \geq C}$ in which there are more single interns than interns who are part of a couple in each hospital. It is easy to generalize our results to the domain of $\mathcal{P}^{S \geq \alpha C}$, in which in every hospital the total weight allocated to single interns is at least $\alpha$ times the total weight allocated to couples ($\forall h\in H: \sum_{s \in S} M_{s,h} \geq \alpha \cdot 2\sum_{c \in C} M_{c_1,h}$), for some positive $\alpha$.

The example given in Theorem~\ref{th:lowerBound} can be modified to have $\frac{2\alpha}{1+\alpha} q + 2$ singles and $\frac{1}{1+\alpha}q - 1$ couples, and the resulting bound will be $\frac{1}{\left(\frac{\alpha}{1+\alpha}\right) \underline{q} + 1}$. The approximation algorithm stays exactly the same, and it provides an approximation of $\frac{1}{\left(\frac{\alpha}{1+\alpha}\right) \underline{q}}$.

On a similar note, we remark that presenting the approximation as depending on the minimal capacity was a matter of choice. A slightly more accurate bound can be formulated in terms of the minimal singles' demand across hospital, i.e., $\min_{h \in H} \sum_{s \in S} M_{s,h} \left( \geq \frac{\alpha}{1+\alpha}\underline{q} \right)$. This bound works much better when for some reason couples focus their demand on larger hospitals.

\subsection{Groups larger than two}
In the Israeli Medical Internship Match, interns were only allowed to register either as singles or as couples. However, other applications may require larger groups to be assigned together. To extend our solution to larger groups it is possible to replace the first stage of the algorithm (the one which rounds up the capacities demanded by the couples and allocates the couples) by the allocation method of \cite{npv2014} (without the singles). Their algorithm ensures an allocation that deviates from the initial allocation by at most $k-1$ seats in each hospital, where $k$ is the size of the largest group. We can then take the remaining capacities and split them among the singles in a similar manner, and the analysis of the upper bound will remain similar (the coefficient however will change according to $k$). Obviously, our two extensions can be combined.

\section{Conclusion} \label{sec:conclusion}

In this thesis we described a recent application of knowledge and research in market design to the problem of allocating interns to internship positions in Israel, we presented a novel technique to perform assignments, and showed that it greatly improved the assignment of Israeli medical graduates to internships, increasing the number of students who received one of their top choices. This method requires the medical students to ``trade'' probabilities to get to different places, and therefore creates a new comparison between different hospitals, based on how much they are desired in the trade. We presented the results of this comparison, and showed that it makes much more sense than the traditional one (namely order the hospitals according to the number of students who ranked them first). Seeing that the new rating makes sense is an evidence that the probabilities in the new lottery are traded correctly.

Our data exhibited several very interesting characteristics that lead us, for example, to recommend (or at least suggest for consideration) merging the two cohorts (of local interns and foreign trained doctors) and assigning them using the same mechanism, since there are likely significant gains to trade. A na\"ive attempt that we made in putting combining these two populations resulted in an improvement for both populations under the new mechanism, and an improvement only for the local interns under RSD (with the foreign trained doctors receiving a worse expected rank). However, it is possible that the MoH will decide that the DNH principle should apply here as well and that the baseline that should be taken is not RSD on the merged population, but rather two independent runs of RSD (one for each population) and running the linear program on top of it.

We expect that decomposing stochastic matrices under small complementarity constraints will rise in other applications. Consider for example a lottery which assigns students to courses. A student needs to care about her probability of getting each course, but would also like the guarantee that two courses will not overlap. While the algorithm we presented here has a worst case approximation ratio $2/\underline{q}$ (where $\underline{q}$ is the minimal capacity in the problem), it behaves much better on simulated and on real data. The reason for this better behavior is that couples do not necessarily concentrate in the small hospitals (and indeed making couples prefer big hospitals improves the performance).

\subsection{Open questions}
\begin{enumerate}
\item What is a more principled way of balancing strategy-proofness with efficiency? We did not pay any extra price for strategy-proofness (in addition to the price for the DNH principle).
\item The students ended up giving no advantages to parents in 2014 and 2015. How would the algorithm behave given a population which consists also of parents? Is there a more principled way to cope with parents?
\item For the assignment problem with couples, what is the lower bound for approximation within the domain $\mathcal{P}^{S \geq C}$ and with elements of $M$ being ``small''?
\item For the assignment problem with couples, our approximation algorithm and analysis assumed that couples must be matched to the same hospital. Can a similar approximation be found when couples can be in different hospitals within the same city?
\end{enumerate}

\newpage
\singlespacing
\clearpage
 \addcontentsline{toc}{section}{Bibliography}
\bibliography{MatchingBibliography}

\begin{thebibliography}{}

\bibitem[\protect\citeauthoryear{Abdulkadiro{\u{g}}lu and
  S{\"{o}}nmez}{Abdulkadiro{\u{g}}lu and S{\"{o}}nmez}{2003}]{as2003}
{\sc Abdulkadiro{\u{g}}lu, A.} {\sc and} {\sc S{\"{o}}nmez, T.} 2003.
\newblock School choice: A mechanism design approach.
\newblock {\em The American Economic Review\/}~{\em 93,\/}~3, 729--747.

\bibitem[\protect\citeauthoryear{Ajtai, Feldman, Hassidim, and Nelson}{Ajtai
  et~al\mbox{.}}{2009}]{ajtai2009sorting}
{\sc Ajtai, M.}, {\sc Feldman, V.}, {\sc Hassidim, A.}, {\sc and} {\sc Nelson,
  J.} 2009.
\newblock Sorting and selection with imprecise comparisons.
\newblock In {\em Automata, Languages and Programming}. Springer, 37--48.

\bibitem[\protect\citeauthoryear{Akbarpour and Nikzad}{Akbarpour and
  Nikzad}{2014}]{an2014}
{\sc Akbarpour, M.} {\sc and} {\sc Nikzad, A.} 2014.
\newblock Approximate random allocation mechanisms.
\newblock available at SSRN 2422777.

\bibitem[\protect\citeauthoryear{Ashlagi, Braverman, and Hassidim}{Ashlagi
  et~al\mbox{.}}{2009}]{ashlagi2009ascending}
{\sc Ashlagi, I.}, {\sc Braverman, M.}, {\sc and} {\sc Hassidim, A.} 2009.
\newblock Ascending unit demand auctions with budget limits.
\newblock Tech. rep., Working paper.

\bibitem[\protect\citeauthoryear{Ashlagi, Braverman, and Hassidim}{Ashlagi
  et~al\mbox{.}}{2014a}]{DBLP:journals/ior/AshlagiBH14}
{\sc Ashlagi, I.}, {\sc Braverman, M.}, {\sc and} {\sc Hassidim, A.} 2014a.
\newblock Stability in large matching markets with complementarities.
\newblock {\em Operations Research\/}~{\em 62,\/}~4, 713--732.

\bibitem[\protect\citeauthoryear{Ashlagi, Braverman, and Hassidim}{Ashlagi
  et~al\mbox{.}}{2014b}]{abh2014}
{\sc Ashlagi, I.}, {\sc Braverman, M.}, {\sc and} {\sc Hassidim, A.} 2014b.
\newblock Stability in large matching markets with complementarities.
\newblock {\em Operations Research\/}.
\newblock forthcoming.

\bibitem[\protect\citeauthoryear{Ashlagi, Braverman, Hassidim, Lavi, and
  Tennenholtz}{Ashlagi et~al\mbox{.}}{2010a}]{ashlagi2010position}
{\sc Ashlagi, I.}, {\sc Braverman, M.}, {\sc Hassidim, A.}, {\sc Lavi, R.},
  {\sc and} {\sc Tennenholtz, M.} 2010a.
\newblock Position auctions with budgets: Existence and uniqueness.
\newblock {\em The BE Journal of Theoretical Economics\/}~{\em 10,\/}~1.

\bibitem[\protect\citeauthoryear{Ashlagi, Braverman, Hassidim, and
  Monderer}{Ashlagi et~al\mbox{.}}{2010b}]{ashlagi2010monotonicity}
{\sc Ashlagi, I.}, {\sc Braverman, M.}, {\sc Hassidim, A.}, {\sc and} {\sc
  Monderer, D.} 2010b.
\newblock Monotonicity and implementability.
\newblock {\em Econometrica\/}~{\em 78,\/}~5, 1749--1772.

\bibitem[\protect\citeauthoryear{Aumann, Dombb, and Hassidim}{Aumann
  et~al\mbox{.}}{2013}]{aumann2013computing}
{\sc Aumann, Y.}, {\sc Dombb, Y.}, {\sc and} {\sc Hassidim, A.} 2013.
\newblock Computing socially-efficient cake divisions.
\newblock In {\em Proceedings of the 2013 international conference on
  Autonomous agents and multi-agent systems}. International Foundation for
  Autonomous Agents and Multiagent Systems, 343--350.

\bibitem[\protect\citeauthoryear{Aumann, Dombb, and Hassidim}{Aumann
  et~al\mbox{.}}{2014}]{aumann2014auctioning}
{\sc Aumann, Y.}, {\sc Dombb, Y.}, {\sc and} {\sc Hassidim, A.} 2014.
\newblock Auctioning a cake: truthful auctions of heterogeneous divisible
  goods.
\newblock In {\em Proceedings of the 2014 international conference on
  Autonomous agents and multi-agent systems}. International Foundation for
  Autonomous Agents and Multiagent Systems, 1045--1052.

\bibitem[\protect\citeauthoryear{Azaria, Hassidim, Kraus, Eshkol, Weintraub,
  and Netanely}{Azaria et~al\mbox{.}}{2013}]{azaria2013movie}
{\sc Azaria, A.}, {\sc Hassidim, A.}, {\sc Kraus, S.}, {\sc Eshkol, A.}, {\sc
  Weintraub, O.}, {\sc and} {\sc Netanely, I.} 2013.
\newblock Movie recommender system for profit maximization.
\newblock In {\em Proceedings of the 7th ACM conference on Recommender
  systems}. ACM, 121--128.

\bibitem[\protect\citeauthoryear{Birkhoff}{Birkhoff}{1946}]{birkhoff1946}
{\sc Birkhoff, G.} 1946.
\newblock Three observations on linear algebra.
\newblock {\em Revi. Univ. Nac. Tucuman, ser A\/}~{\em 5}, 147--151.

\bibitem[\protect\citeauthoryear{Bogomolnaia and Moulin}{Bogomolnaia and
  Moulin}{2001}]{bm2001}
{\sc Bogomolnaia, A.} {\sc and} {\sc Moulin, H.} 2001.
\newblock A new solution to the random assignment problem.
\newblock {\em Journal of Economic Theory\/}~{\em 100,\/}~2, 295--328.

\bibitem[\protect\citeauthoryear{Bronfman, Hassidim, Afek, Romm, Sherberk,
  Hassidim, and Massler}{Bronfman et~al\mbox{.}}{2015}]{bharshm2015}
{\sc Bronfman, S.}, {\sc Hassidim, A.}, {\sc Afek, A.}, {\sc Romm, A.}, {\sc
  Sherberk, R.}, {\sc Hassidim, A.}, {\sc and} {\sc Massler, A.} 2015.
\newblock Assigning israeli medical graduates to internships.
\newblock {\em Israel Journal of Health Policy Research\/}.

\bibitem[\protect\citeauthoryear{Budish}{Budish}{2011}]{budish2011}
{\sc Budish, E.} 2011.
\newblock The combinatorial assignment problem: Approximate competitive
  equilibrium from equal incomes.
\newblock {\em Journal of Political Economy\/}~{\em 119,\/}~6, 1061--1103.

\bibitem[\protect\citeauthoryear{Budish, Che, Kojima, and Milgrom}{Budish
  et~al\mbox{.}}{2013}]{bckm2013}
{\sc Budish, E.}, {\sc Che, Y.-K.}, {\sc Kojima, F.}, {\sc and} {\sc Milgrom,
  P.} 2013.
\newblock Designing random allocation mechanisms: Theory and applications.
\newblock {\em The American Economic Review\/}~{\em 103,\/}~2, 585--623.

\bibitem[\protect\citeauthoryear{Chen, Hassidim, and Micali}{Chen
  et~al\mbox{.}}{2010}]{chen2010robust}
{\sc Chen, J.}, {\sc Hassidim, A.}, {\sc and} {\sc Micali, S.} 2010.
\newblock Robust perfect revenue from perfectly informed players.
\newblock In {\em ICS}. 94--105.

\bibitem[\protect\citeauthoryear{Cook, Fabrikant, and Hassidim}{Cook
  et~al\mbox{.}}{2013}]{cook2013grow}
{\sc Cook, J.}, {\sc Fabrikant, A.}, {\sc and} {\sc Hassidim, A.} 2013.
\newblock How to grow more pairs: suggesting review targets for
  comparison-friendly review ecosystems.
\newblock In {\em 22nd International World Wide Web Conference, {WWW} '13, Rio
  de Janeiro, Brazil, May 13-17, 2013}. 237--248.

\bibitem[\protect\citeauthoryear{Danna, Hassidim, Kaplan, Kumar, Mansour, Raz,
  and Segalov}{Danna et~al\mbox{.}}{2012}]{danna2012upward}
{\sc Danna, E.}, {\sc Hassidim, A.}, {\sc Kaplan, H.}, {\sc Kumar, A.}, {\sc
  Mansour, Y.}, {\sc Raz, D.}, {\sc and} {\sc Segalov, M.} 2012.
\newblock Upward max min fairness.
\newblock In {\em INFOCOM, 2012 Proceedings IEEE}. IEEE, 837--845.

\bibitem[\protect\citeauthoryear{Farhi, Gosset, Hassidim, Lutomirski, and
  Shor}{Farhi et~al\mbox{.}}{2012}]{farhi2012quantum}
{\sc Farhi, E.}, {\sc Gosset, D.}, {\sc Hassidim, A.}, {\sc Lutomirski, A.},
  {\sc and} {\sc Shor, P.~W.} 2012.
\newblock Quantum money from knots.
\newblock In {\em Innovations in Theoretical Computer Science 2012, Cambridge,
  MA, USA, January 8-10, 2012}. 276--289.

\bibitem[\protect\citeauthoryear{Featherstone}{Featherstone}{2014}]{featherstone2014}
{\sc Featherstone, C.~R.} 2014.
\newblock Rank efficiency: Investigating a widespread ordinal welfare
  criterion.
\newblock mimeo.

\bibitem[\protect\citeauthoryear{Hajaj, Dickerson, Hassidim, Sandholm, and
  Sarne}{Hajaj et~al\mbox{.}}{2015}]{hajaj2015strategy}
{\sc Hajaj, C.}, {\sc Dickerson, J.~P.}, {\sc Hassidim, A.}, {\sc Sandholm,
  T.}, {\sc and} {\sc Sarne, D.} 2015.
\newblock Strategy-proof and efficient kidney exchange using a credit
  mechanism.
\newblock In {\em AAAI Conference on Artificial Intelligence (AAAI)}.

\bibitem[\protect\citeauthoryear{Hassidim, Kaplan, Mansour, and Nisan}{Hassidim
  et~al\mbox{.}}{2011}]{hassidim2011non}
{\sc Hassidim, A.}, {\sc Kaplan, H.}, {\sc Mansour, Y.}, {\sc and} {\sc Nisan,
  N.} 2011.
\newblock Non-price equilibria in markets of discrete goods.
\newblock In {\em Proceedings of the 12th ACM conference on Electronic
  commerce}. ACM, 295--296.

\bibitem[\protect\citeauthoryear{Hassidim, Kaplan, Mansour, and Nisan}{Hassidim
  et~al\mbox{.}}{2012}]{DBLP:conf/wine/HassidimKMN12}
{\sc Hassidim, A.}, {\sc Kaplan, H.}, {\sc Mansour, Y.}, {\sc and} {\sc Nisan,
  N.} 2012.
\newblock The {AND-OR} game: Equilibrium characterization - (working paper).
\newblock In {\em Internet and Network Economics - 8th International Workshop,
  {WINE} 2012, Liverpool, UK, December 10-12, 2012. Proceedings}. 561.

\bibitem[\protect\citeauthoryear{Hassidim, Mansour, and Vardi}{Hassidim
  et~al\mbox{.}}{2014}]{hassidim2014local}
{\sc Hassidim, A.}, {\sc Mansour, Y.}, {\sc and} {\sc Vardi, S.} 2014.
\newblock Local computation mechanism design.
\newblock In {\em Proceedings of the fifteenth ACM conference on Economics and
  computation}. ACM, 601--616.

\bibitem[\protect\citeauthoryear{Hassidim, Raz, Segalov, and Shaqed}{Hassidim
  et~al\mbox{.}}{2013}]{hassidim2013network}
{\sc Hassidim, A.}, {\sc Raz, D.}, {\sc Segalov, M.}, {\sc and} {\sc Shaqed,
  A.} 2013.
\newblock Network utilization: The flow view.
\newblock In {\em INFOCOM, 2013 Proceedings IEEE}. IEEE, 1429--1437.

\bibitem[\protect\citeauthoryear{Hassidim and Romm}{Hassidim and
  Romm}{2014}]{hassidim2014approximate}
{\sc Hassidim, A.} {\sc and} {\sc Romm, A.} 2014.
\newblock An approximate “law of one price” in random assignment games.
\newblock {\em arXiv preprint arXiv:1404.6103\/}.

\bibitem[\protect\citeauthoryear{Hylland and Zeckhauser}{Hylland and
  Zeckhauser}{1979}]{hz1979}
{\sc Hylland, A.} {\sc and} {\sc Zeckhauser, R.} 1979.
\newblock The efficient allocation of individuals to positions.
\newblock {\em The Journal of Political Economy\/}, 293--314.

\bibitem[\protect\citeauthoryear{I.Holyer}{I.Holyer}{1981}]{holyer1981}
{\sc I.Holyer}. 1981.
\newblock The np-completeness of some edge-partition problems.
\newblock {\em SIAM J. Comput.\/}~{\em 10,\/}~4, 713--717.

\bibitem[\protect\citeauthoryear{Kantorovich}{Kantorovich}{1940}]{Kantorovich}
{\sc Kantorovich, L.~V.} 1940.
\newblock A new method of solving of some classes of extremal problems.
\newblock In {\em Dokl. Akad. Nauk SSSR}. Vol.~28. 211--214.

\bibitem[\protect\citeauthoryear{Kojima, Pathak, and Roth}{Kojima
  et~al\mbox{.}}{2013}]{kpr2013}
{\sc Kojima, F.}, {\sc Pathak, P.~A.}, {\sc and} {\sc Roth, A.~E.} 2013.
\newblock Matching with couples: Stability and incentives in large markets.
\newblock {\em The Quarterly Journal of Economics\/}~{\em 128,\/}~4,
  1585--1632.

\bibitem[\protect\citeauthoryear{Liu and Pycia}{Liu and
  Pycia}{2013}]{liu2013ordinal}
{\sc Liu, Q.} {\sc and} {\sc Pycia, M.} 2013.
\newblock Ordinal efficiency, fairness, and incentives in large markets.
\newblock {\em Fairness, and Incentives in Large Markets (September 2013)\/}.

\bibitem[\protect\citeauthoryear{Lutomirski, Aaronson, Farhi, Gosset, Kelner,
  Hassidim, and Shor}{Lutomirski et~al\mbox{.}}{}]{lutomirskibreaking}
{\sc Lutomirski, A.}, {\sc Aaronson, S.}, {\sc Farhi, E.}, {\sc Gosset, D.},
  {\sc Kelner, J.}, {\sc Hassidim, A.}, {\sc and} {\sc Shor, P.}
\newblock Breaking and making quantum money: Toward a new quantum cryptographic
  protocol.

\bibitem[\protect\citeauthoryear{Nguyen, Peivandi, and Vohra}{Nguyen
  et~al\mbox{.}}{2014}]{npv2014}
{\sc Nguyen, T.}, {\sc Peivandi, A.}, {\sc and} {\sc Vohra, R.} 2014.
\newblock One-sided matching with limited complementarities.
\newblock PIER Working Paper.

\bibitem[\protect\citeauthoryear{Nguyen and Vohra}{Nguyen and
  Vohra}{2014}]{nv2014}
{\sc Nguyen, T.} {\sc and} {\sc Vohra, R.} 2014.
\newblock Near feasible stable matchings with complementarities.
\newblock PIER Working Paper.

\bibitem[\protect\citeauthoryear{Plummer and Lov{\'a}sz}{Plummer and
  Lov{\'a}sz}{1986}]{pml1986}
{\sc Plummer, M.~D.} {\sc and} {\sc Lov{\'a}sz, L.} 1986.
\newblock {\em Matching theory}.
\newblock Elsevier.

\bibitem[\protect\citeauthoryear{Segal-Halevi, Hassidim, and
  Aumann}{Segal-Halevi et~al\mbox{.}}{2014}]{segal2014fair}
{\sc Segal-Halevi, E.}, {\sc Hassidim, A.}, {\sc and} {\sc Aumann, Y.} 2014.
\newblock Fair and square: Cake-cutting in two dimensions.
\newblock {\em arXiv preprint arXiv:1409.4511\/}.

\bibitem[\protect\citeauthoryear{Segal-Halevi, Hassidim, and
  Aumann}{Segal-Halevi et~al\mbox{.}}{2015a}]{segal2015envy}
{\sc Segal-Halevi, E.}, {\sc Hassidim, A.}, {\sc and} {\sc Aumann, Y.} 2015a.
\newblock Envy-free cake-cutting in two dimensions.
\newblock In {\em AAAI Conference on Artificial Intelligence (AAAI)}.

\bibitem[\protect\citeauthoryear{Segal-Halevi, Hassidim, and
  Aumann}{Segal-Halevi et~al\mbox{.}}{2015b}]{segal2015waste}
{\sc Segal-Halevi, E.}, {\sc Hassidim, A.}, {\sc and} {\sc Aumann, Y.} 2015b.
\newblock Waste makes haste: Bounded time protocols for envy-free cake cutting
  with free disposal.
\newblock In {\em Proceedings of the 2015 International Conference on
  Autonomous Agents and Multiagent Systems}. International Foundation for
  Autonomous Agents and Multiagent Systems, 901--908.

\bibitem[\protect\citeauthoryear{Shapley and Scarf}{Shapley and
  Scarf}{1974}]{ss1974}
{\sc Shapley, L.} {\sc and} {\sc Scarf, H.} 1974.
\newblock On cores and indivisibility.
\newblock {\em Journal of mathematical economics\/}~{\em 1,\/}~1, 23--37.

\bibitem[\protect\citeauthoryear{Sina, Hazon, Hassidim, and Kraus}{Sina
  et~al\mbox{.}}{2015}]{sina2015adapting}
{\sc Sina, S.}, {\sc Hazon, N.}, {\sc Hassidim, A.}, {\sc and} {\sc Kraus, S.}
  2015.
\newblock Adapting the social network to affect elections.
\newblock In {\em Proceedings of the 2015 International Conference on
  Autonomous Agents and Multiagent Systems}. International Foundation for
  Autonomous Agents and Multiagent Systems, 705--713.

\bibitem[\protect\citeauthoryear{Sinkhorn and Knopp}{Sinkhorn and
  Knopp}{1967}]{sk1967}
{\sc Sinkhorn, R.} {\sc and} {\sc Knopp, P.} 1967.
\newblock Concerning nonnegative matrices and doubly stochastic matrices.
\newblock {\em Pacific Journal of Mathematics\/}~{\em 21,\/}~2, 343--348.

\bibitem[\protect\citeauthoryear{Sofer, Sarne, and Hassidim}{Sofer
  et~al\mbox{.}}{2012}]{sofer2012negotiation}
{\sc Sofer, I.}, {\sc Sarne, D.}, {\sc and} {\sc Hassidim, A.} 2012.
\newblock Negotiation in exploration-based environment.
\newblock In {\em Twenty-Sixth AAAI Conference on Artificial Intelligence}.

\bibitem[\protect\citeauthoryear{Von~Neumann}{Von~Neumann}{1953}]{vonneumann1953}
{\sc Von~Neumann, J.} 1953.
\newblock A certain zero-sum two-person game equivalent to the optimal
  assignment problem.
\newblock {\em Contributions to the Theory of Games\/}~{\em 2}, 5--12.

\end{thebibliography}
\newpage

\appendix

\section{Linear programming technical explanation}
\label{app:A}
To represent the \textit{individual rationality constraint}, Let $m$ be the number of hospitals, we define Student~$i$'s happiness prior to the trading stage as
\begin{equation}
		\label{eq:happinessBefore}
		 h_i=p_{i,1}{m^2}+p_{i,2}{(m-1)^2}+p_{i,3}{(m-2)^2}+\ldots{}+p_{i,m}{1^2},
	\end{equation}
	i.e., a function that represents students' strong preference to get to their top ranked hospitals. If the probabilities we get after the trading stage are $\hat{p}_{i,k}$ (as defined in Eq.~\ref{eq:prob}), then we define the happiness following the optimization as
	\begin{equation}
		\label{eq:happinessAfter}
		 \hat{h}_i=\hat{p}_{i,1}{m^2}+\hat{p}_{i,2}{(m-1)^2}+\hat{p}_{i,3}{(m-2)^2}+...+\hat{p}_{i,m}{1^2}.
	\end{equation}
Now our constraints are $\hat{h}_i \geq h_i$ for every student $i$.

Regarding the target function, as described in the Methods section, we want to maximize the total satisfaction of the students after trading. Letting $n$ denote the number of students (in 2014, $n = 496$), our optimization goal is going to be
\begin{equation}
		\label{eq:target2}
		\max  \sum_{i=1}^n{\hat{h}_i}=\hat{h}_1+\hat{h}_2+...+\hat{h}_n,
\end{equation}
that is, maximizing the sum of happiness for all interns. This target function, as well as the definition of happiness in Eq.~\ref{eq:happinessBefore}, was chosen following a survey filled by approximately 70 interns and 6th year medical students.
We note that using similar target functions and making small changes in the weights used to define happiness did not have a profound effect on the statistics of the assignment.

\section{Lower Bound for small probabilities} \label{apndx:small_probs}
Define an example with
\begin{align*}
H &= \left\{h_1,\ldots,h_{2m},h'_1,\ldots,h'_{2m}\right\}, \\
S &= \{s_1,\ldots,s_{2m(2k+1)}\}, \text{and} \\
C &= \{c_1,\ldots,c_{(2k+1)m}\},
\end{align*}
where $m$ and $k$ are integers. Consider the target matrix $M\in \mathcal{P}^{S \geq C}$ described in Table~II. Under this stochastic assignment matrix each single intern in $S$ gets a probability $\frac{1}{2m}$ to be assigned to each hospital in $\{h'_1,\ldots,h'_{2m}\}$, and each couple in $C$ gets a probability of $\frac{1}{2m}$ to be assigned to each hospital in $\{h_1,\ldots,h_{2m}\}$. One can verify that the capacity of every hospital is $2k+1$.

In each deterministic assignment, the couples occupy at most $4mk$ slots in the hospitals in $\{h_1,\ldots,h_{2m}\}$, hence the singletons occupy at least $2m$ slots in $\{h_1,\ldots,h_{2m}\}$. This is true for every assignment, hence true for their convex combination as well, meaning that some single
intern spends at least $\frac{2m}{(2m)(2k+1)}=\frac{1}{2k+1}$ of his probability in the hospitals in $\{h_1,\ldots,h_{2m}\}$. Thus, the approximation cannot be better than $\frac{2}{2k+1} = \left(\frac{2}{\underline{q}}\right)$.

\begin{table}[htbp]
	\caption{Target matrix for small probabilities example}
	\centering
		\begin{tabu}{|c|c|c|c|c|c|c|}
			\hline
			\textbf{} & \textbf{$h_1$} & \textbf{$\cdots$} & \textbf{$h_{2m}$} & \textbf{$h'_1$} & \textbf{$\cdots$} & \textbf{$h'_{2m}$} \\ \hline
			
			$s_1$ & $0$ & $\cdots$ & $0$ & $1/(2m)$ & $\cdots$ & $1/(2m)$ \\ \hline
			\vdots & $0$ & $\cdots$ & $0$ & $1/(2m)$ & $\cdots$ & $1/(2m)$ \\ \hline
			$s_{2m(2k+1)}$ & $0$ & $\cdots$ & $0$ & $1/(2m)$ & $\cdots$ & $1/(2m)$ \\ \tabucline[1pt]{-}
			
			$c_{1,1}$ & $1/(2m)$ & $\cdots$ & $1/(2m)$ & $0$ & $\cdots$ & $0$ \\ \hline
			$c_{1,2}$ & $1/(2m)$ & $\cdots$ & $1/(2m)$ & $0$ & $\cdots$ & $0$ \\ \hline
			\vdots & $1/(2m)$ & $\cdots$ & $1/(2m)$ & $0$ & $\cdots$ & $0$ \\ \hline
			$c_{{m(2k+1)},1}$ & $1/(2m)$ & $\cdots$ & $1/(2m)$ & $0$ & $\cdots$ & $0$ \\ \hline
			$c_{{m(2k+1)},2}$ & $1/(2m)$ & $\cdots$ & $1/(2m)$ & $0$ & $\cdots$ & $0$ \\ \tabucline[1pt]{-}
		
		\end{tabu}
		\label{tab:smallProbsMatrix}
		
\end{table}

\end{document}